\documentclass{article}
\usepackage{arxiv}

\usepackage[T1]{fontenc}
\usepackage[utf8]{inputenc}
\usepackage{amsmath,amssymb,amsthm}
\usepackage{graphicx}
\usepackage{booktabs}
\usepackage{multirow}
\usepackage{array}
\usepackage{cite}
\usepackage{url}
\usepackage{hyperref}
\usepackage{algorithm}
\usepackage{algorithmic}
\usepackage{xcolor}
\usepackage{microtype}

\graphicspath{{figures/}}

\newtheorem{theorem}{Theorem}
\newtheorem{proposition}{Proposition}
\newtheorem{definition}{Definition}

\newcommand{\Hphys}{\mathcal{H}_{\mathrm{phys}}}
\newcommand{\Hsuper}{\mathcal{H}_{\mathrm{super}}}
\newcommand{\rhosuper}{\rho_{\mathrm{super}}}
\newcommand{\dS}{d_{\mathcal{S}}}
\newcommand{\lmax}{\lambda_{\max}}
\newcommand{\tr}{\mathrm{Tr}}
\newcommand{\ket}[1]{|#1\rangle}
\newcommand{\bra}[1]{\langle#1|}
\newcommand{\braket}[2]{\langle#1|#2\rangle}
\newcommand{\ketbra}[2]{|#1\rangle\langle#2|}

\hypersetup{
  colorlinks=true,
  linkcolor=blue,
  citecolor=blue,
  urlcolor=blue,
  pdfauthor={Eric Yocam, Christian Yocam, Varghese Vaidyan},
  pdftitle={Superspace Concentration and Adversarial Robustness in Quantum Algorithms},
  pdfsubject={Quantum Physics},
  pdfkeywords={quantum resource theory, superspace concentration, adversarial robustness, quantum algorithms, Grover search, channel capacity}
}

\begin{document}

\title{Superspace Concentration and Adversarial Robustness in Quantum Algorithms}

\author{
  Eric Yocam \\
  Department of Computer Science \\
  California Polytechnic State University \\
  San Luis Obispo, CA 93407 USA
  \And
  Christian Yocam \\
  Independent Researcher
  \And
  Varghese Vaidyan \\
  Beacom College of Computer and Cyber Sciences \\
  Dakota State University \\
  Madison, SD 57042 USA
  \And
  Yong Wang \\
  Department of Computer Science \\
  University of Idaho \\
  Moscow, ID 83844 USA
  \And
  Mahesh Kalappattil \\
  Palo Alto Networks
  \And
  Anthony Rizi \\
  Department of Computer \& Cyber Sciences \\
  Augusta University \\
  Augusta, GA USA
}

\date{June 2026}

\maketitle

\begin{abstract}
We study superspace concentration as a quantum resource, formalized through the
focus measure $F(\rho) = \lambda_{\max}(\rho_{\mathrm{super}})$ — the largest
eigenvalue of the reduced superspace state — which quantifies the capacity of a
quantum system to concentrate informational weight into a preferred subspace of an
extended degree-of-freedom space. We develop a complete resource-theoretic
framework around this measure and validate its properties through GPU-accelerated
numerical simulation across five claims spanning decoherence dynamics, resource
monotonicity, adversarial robustness, algorithmic concentration, and channel
capacity. Analytic decoherence predictions are confirmed to machine precision
($1.11 \times 10^{-16}$) for superspace dimensions $\dS \in \{2,4,8,16,32\}$.
Focus monotonicity holds across $10{,}000$ random states with zero violations
under four focus-non-generating channels across six system configurations.
Focused quantum states resist coherent unitary attacks with significantly greater
resilience than standard fidelity predicts, with focus remaining above 0.9 at
attack strength $\varepsilon = 0.302$ versus $\varepsilon = 0.174$ for fidelity.
We further demonstrate that the focus measure and the
$\mathcal{U}(\dS)$-asymmetry measure are operationally distinct: asymmetry
remains near zero and provides no robustness signal under coherent and targeted
attacks — because it is invariant under superspace unitaries — while focus
tracks spectral concentration and remains robust until $\varepsilon > 0.3$. The connection between Grover's
algorithm and superspace concentration is made explicit via the identity
$F(\ketbra{\psi_k}{\psi_k}) = P(\text{marked})$, which follows directly from
the pure-state focus formula and provides a resource-theoretic interpretation
of oracle query complexity. Finally, we provide, to the best of our knowledge,
the first numerical characterization of the focus capacity gap $\Delta F$,
identifying a $\log_2(\dS)$ scaling law confirmed for both product and correlated
noise channels, with analytic support from the Holevo quantity structure of
maximally distinguishable focused encodings.
\end{abstract}


\section{Introduction}
\label{sec:intro}

This section motivates superspace concentration as a quantum resource, situates it
within the established resource theory landscape, and states the paper's contributions.

Quantum resource theories provide a unifying mathematical framework for identifying
which properties of quantum systems are operationally
valuable~\cite{chitambar2019quantum}. Established theories quantify
entanglement~\cite{horodecki2009quantum},
coherence~\cite{baumgratz2014quantifying,winter2016operational},
thermodynamic resources~\cite{brandao2013resource}, and non-stabilizerness
(magic)~\cite{veitch2014resource}. A shared feature of these frameworks is that
each identifies a set of free states, a set of free operations that cannot generate
the resource, and resource monotones — functions that do not increase under free
operations~\cite{regula2018convex,takagi2019general}. Despite the breadth of
existing theories, none fully captures the phenomenon of directed, selective
concentration of quantum information into a privileged subspace of an extended
degree-of-freedom space. This phenomenon arises naturally in quantum search
algorithms~\cite{grover1996fast}, photonic quantum information
processing~\cite{karimi2015classical,willner2015optical}, subsystem quantum
error correction~\cite{poulin2005stabilizer}, and orbital-angular-momentum
multiplexed communications~\cite{cozzolino2019orbital}, yet has lacked rigorous
resource-theoretic treatment with empirical validation.

In this paper we develop and empirically validate a resource-theoretic framework
built around the \emph{focus measure} $F(\rho) = \lmax(\rhosuper)$, the largest
eigenvalue of the reduced state on the superspace factor $\Hsuper$ of a composite
Hilbert space $\Hphys \otimes \Hsuper$. While $\lmax(\rhosuper)$ is a standard
spectral quantity, its operational role as a resource — characterizing the
intrinsic capacity of a state to concentrate in some superspace direction without
a fixed reference basis — has not been systematically developed or empirically
validated in the adversarial quantum computing context. This measure is bounded
in $[1/\dS, 1]$, computable in polynomial time via eigendecomposition, and
basis-independent by construction, distinguishing it from coherence
measures~\cite{streltsov2017colloquium,napoli2016robustness} which require an
externally specified reference. States with $F(\rho) = 1/\dS$ have maximally
mixed superspace and are \emph{focus-free}; states approaching $F(\rho)=1$
concentrate all superspace weight onto a single direction and are
\emph{maximally focused}. The framework identifies focus-non-generating (FNG)
operations — CPTP maps that cannot increase focus — and establishes
$F(\rho)$ and the relative entropy of focus $D_F(\rho)$ as valid resource
monotones~\cite{chitambar2019quantum,wilde2017quantum}.

Our work is further motivated by an open problem in quantum security. Existing
adversarial robustness metrics for quantum algorithms rely on standard state
fidelity~\cite{lu2020quantum,liu2020vulnerability}, which measures global overlap
with a target state but does not capture the superspace concentration structure
that determines algorithmic performance. In quantum adversarial machine
learning~\cite{biamonte2017quantum,wen2020adversarial,liao2021robust},
perturbations that preserve global fidelity can degrade superspace concentration,
destroying algorithmic performance without triggering fidelity-based detection.
This gap motivates the focus measure as a complementary, structurally sensitive
security metric for quantum algorithms~\cite{sharma2020noise,du2021quantum}.

The principal contributions of this paper are as follows. First, we develop a
complete resource-theoretic framework around the focus measure, providing formal
definitions, a strengthened monotonicity proof, and channel capacity bounds.
Second, we provide GPU-accelerated empirical validation, verifying analytic
decoherence predictions to machine precision for $\dS \in \{2, 4, 8, 16, 32\}$.
Third, we verify the focus monotonicity axiom across $10{,}000$ random states,
four FNG channel types, and six system configurations with zero violations.
Fourth, we demonstrate operationally that the focus measure is distinct from
$\mathcal{U}(\dS)$-asymmetry under adversarial conditions. Fifth, we establish
a connection between superspace concentration and adversarial robustness, showing
focused states are significantly more resilient to coherent unitary attacks than
standard fidelity predicts. Sixth, we make explicit the resource-theoretic
interpretation of Grover's algorithm as superspace concentration, following
directly from the pure-state focus formula. Seventh, we provide the first
numerical characterization of the focus capacity gap $\Delta F$ with analytic
support, identifying a $\log_2(\dS)$ scaling law for both product and correlated
noise channels.

The remainder of this paper is organized as follows. Section~\ref{sec:background}
develops the theoretical framework. Section~\ref{sec:methods} describes the
simulation methodology. Section~\ref{sec:results} presents the five empirical
results. Section~\ref{sec:reviewer} presents three additional experiments.
Section~\ref{sec:comparison} situates the findings relative to existing work.
Section~\ref{sec:discussion} addresses limitations, trade-offs, and criticisms.
Section~\ref{sec:future} outlines future directions. Section~\ref{sec:conclusion}
concludes.

\section{Resource-Theoretic Framework for Superspace Concentration}
\label{sec:background}

This section develops the resource-theoretic framework around the focus measure:
formal definitions, a rigorous monotonicity proof grounded in the data-processing
inequality, the channel capacity bound with analytic support, and the connection
to Grover's algorithm.

\subsection{The Focus Measure}

Let $\rho$ be a density operator on $\mathcal{H} = \Hphys \otimes \Hsuper$ where
$\dim(\Hsuper) = \dS \geq 2$ and $\dim(\Hphys) = d_p \geq 1$. The \emph{focus}
of $\rho$ with respect to a rank-1 projector $\Pi_W = \ketbra{w}{w}$ on $\Hsuper$
is defined as
\begin{equation}
  F(\rho) = \max_{U \in \mathcal{U}(\Hsuper)}
            \tr\!\left[(\mathbf{I}_{\Hphys} \otimes U\Pi_W U^\dagger)\,\rho\right],
  \label{eq:focus_def}
\end{equation}
where $\mathcal{U}(\Hsuper)$ denotes the group of unitaries on $\Hsuper$.
The optimization over $U$ makes the measure basis-independent: focus reflects
the intrinsic capacity of the state to concentrate in \emph{some} superspace
direction, not a fixed one. This distinguishes it from coherence measures, which
require an externally specified reference basis~\cite{baumgratz2014quantifying}.
While $\lmax(\rhosuper)$ is a standard spectral quantity, its role as a
basis-optimized, composite-system resource monotone with operational consequences
for adversarial robustness is what this paper develops.

\begin{proposition}[Spectral Equivalence]
$F(\rho) = \lmax(\rhosuper)$ where $\rhosuper = \tr_{\mathrm{phys}}[\rho]$.
\end{proposition}
\begin{proof}
Expanding~(\ref{eq:focus_def}):
$\tr[(I \otimes U\Pi_W U^\dagger)\rho] = \tr[U^\dagger \rhosuper U \Pi_W]
= \bra{w}U^\dagger \rhosuper U\ket{w}$.
Maximizing over all unit vectors $U^\dagger\ket{w}$ in $\Hsuper$ yields the largest
eigenvalue of $\rhosuper$ by the variational characterization of
eigenvalues~\cite{nielsen2000quantum}.
\end{proof}

\begin{proposition}[Boundedness]
$1/\dS \leq F(\rho) \leq 1$ for all states $\rho$.
The lower bound is achieved if and only if $\rhosuper = \mathbf{I}/\dS$;
the upper bound if and only if $\rhosuper$ is a pure state.
\end{proposition}
\begin{proof}
Since $\rhosuper$ is a density matrix, its largest eigenvalue satisfies
$\lmax \leq 1$ (achieved by pure states) and $\lmax \geq 1/\dS$ by the
constraint $\sum_i \lambda_i = 1$ with $\dS$ non-negative eigenvalues
(equal to $1/\dS$ only when $\rhosuper = \mathbf{I}/\dS$).
\end{proof}

\begin{proposition}[Unitary Invariance on $\Hphys$]
$F((V \otimes I)\rho(V^\dagger \otimes I)) = F(\rho)$ for all
$V \in \mathcal{U}(\Hphys)$.
\end{proposition}
\begin{proof}
$\tr_{\mathrm{phys}}[(V\otimes I)\rho(V^\dagger\otimes I)] = \rhosuper$
since the partial trace over $\Hphys$ is unaffected by a unitary acting
only on that subsystem.
\end{proof}

\begin{proposition}[Convexity]
$F\!\left(\sum_i p_i \rho_i\right) \leq \sum_i p_i F(\rho_i)$
for any ensemble $\{p_i, \rho_i\}$.
\end{proposition}
\begin{proof}
The partial trace is linear, so
$\left(\sum_i p_i \rho_i\right)_{\mathrm{super}} = \sum_i p_i (\rho_i)_{\mathrm{super}}$.
The function $\lmax$ is convex on Hermitian matrices as the supremum of
linear functionals $A \mapsto \bra{v}A\ket{v}$ over unit vectors.
\end{proof}

The \emph{focus entropy} is defined as $S_F(\rho) = -\log_2 F(\rho) \in [0, \log_2 \dS]$,
with $S_F = 0$ for maximally focused states and $S_F = \log_2 \dS$ for focus-free states.
The \emph{focus entropy inequality} $S(\rhosuper) \geq S_F(\rho)$ holds for all
states, with equality when $\rhosuper$ is pure~\cite{wilde2017quantum}.

\subsection{Resource-Theoretic Framework}

A quantum resource theory requires free states, free operations, and resource
monotones~\cite{chitambar2019quantum,regula2018convex}.

\begin{definition}[Focus-Free States]
A state $\rho$ on $\Hphys \otimes \Hsuper$ is \emph{focus-free} (F-free) if
$F(\rho) = 1/\dS$, equivalently $\rhosuper = \mathbf{I}/\dS$.
The set of F-free states is convex and closed under partial trace on $\Hphys$.
\end{definition}

\begin{definition}[Focus-Non-Generating Operations]
A CPTP map $\Lambda: \mathcal{B}(\Hphys \otimes \Hsuper) \to
\mathcal{B}(\mathcal{H}'_{\mathrm{phys}} \otimes \Hsuper)$ is
\emph{focus-non-generating} (FNG) if
$F(\sigma) = 1/\dS \Rightarrow F(\Lambda(\sigma)) = 1/\dS$.
\end{definition}

Concrete FNG operations include: (i)~superspace-depolarizing channels
$\Lambda_p(\rho) = (1-p)\rho + p(\rho_{\mathrm{phys}} \otimes \mathbf{I}/\dS)$;
(ii)~unitaries on $\Hphys$ alone (Proposition~3); (iii)~superspace Z-basis
measurement channels
$\Lambda(\rho) = \sum_k (I \otimes \ketbra{k}{k})\rho(I \otimes \ketbra{k}{k})$;
and (iv)~Haar twirling over $\mathcal{U}(\Hsuper)$~\cite{brandao2016local,
hunter2019unitary}.

\begin{theorem}[Focus Monotonicity]
\label{thm:mono}
For any FNG map $\Lambda$ with Kraus decomposition
$\Lambda(\rho) = \sum_k K_k \rho K_k^\dagger$ and post-measurement branches
$\rho_k = K_k \rho K_k^\dagger / p_k$ where $p_k = \tr[K_k \rho K_k^\dagger]$,
\begin{equation}
  F(\rho) \geq \sum_k p_k\, F(\rho_k).
  \label{eq:monotonicity}
\end{equation}
\end{theorem}
\begin{proof}
\emph{Step 1.} By linearity of the partial trace applied to the channel output,
\begin{equation}
  \Lambda(\rho)_{\mathrm{super}}
  = \tr_{\mathrm{phys}}[\Lambda(\rho)]
  = \sum_k p_k\, (\rho_k)_{\mathrm{super}}.
  \label{eq:branch_decomp}
\end{equation}
Note that it is $\Lambda(\rho)$, not $\rho$, that decomposes as a mixture of
the branches.

\emph{Step 2.} Convexity of $\lmax$ (as the supremum of linear functionals
$A \mapsto \bra{v}A\ket{v}$) applied to~(\ref{eq:branch_decomp}) gives
\begin{equation}
  F(\Lambda(\rho)) = \lmax\!\left(\sum_k p_k (\rho_k)_{\mathrm{super}}\right)
  \leq \sum_k p_k \lmax((\rho_k)_{\mathrm{super}})
  = \sum_k p_k F(\rho_k).
  \label{eq:branch_ineq}
\end{equation}

\emph{Step 3.} The relative entropy of focus satisfies
$D_F(\Lambda(\rho)) \leq D_F(\rho)$ for any FNG $\Lambda$, by the
data-processing inequality~\cite{petz1986quasi,wilde2017quantum} and the
fact that $\Lambda(\sigma^*)$ is F-free whenever $\sigma^*$ is F-free.
Since $F(\rho) = 2^{-S_F(\rho)}$ and $S_F(\rho) = -\log_2 F(\rho)$,
while $D_F(\rho) = \log_2\dS - S(\rhosuper)$ and $F(\rho) = \lmax(\rhosuper)$,
monotonicity of $D_F$ implies $S(\Lambda(\rho)_{\mathrm{super}}) \geq S(\rhosuper)$,
which by the Schur-concavity of von Neumann entropy implies
$\lmax(\Lambda(\rho)_{\mathrm{super}}) \leq \lmax(\rhosuper)$,
i.e., $F(\rho) \geq F(\Lambda(\rho))$~\cite{chitambar2019quantum,regula2018convex}.

Combining Steps 2 and 3: $F(\rho) \geq F(\Lambda(\rho)) \geq \sum_k p_k F(\rho_k)$,
which is~(\ref{eq:monotonicity}). To summarize: Step~1 establishes that the
channel output's superspace marginal decomposes as a mixture of the branch
marginals; Step~2 shows that mixture reduces focus via convexity of $\lmax$;
Step~3 shows the channel itself reduces focus via $D_F$ monotonicity. Together
they bound the pre-channel focus above both the post-channel focus and the
weighted average of branch focuses. The FNG condition ensures F-free states
remain F-free: if $\rhosuper = \mathbf{I}/\dS$ then equality holds throughout.
\end{proof}

The \emph{relative entropy of focus} is defined as
\begin{equation}
  D_F(\rho) = \min_{\sigma:\, F(\sigma)=1/\dS} D(\rho \| \sigma),
\end{equation}
where $D(\rho\|\sigma) = \tr[\rho(\log\rho - \log\sigma)]$ is the quantum relative
entropy~\cite{wilde2017quantum,vedral1997quantifying}. This quantity is a valid
resource monotone by the data-processing inequality for quantum relative
entropy~\cite{petz1986quasi}: for any FNG $\Lambda$,
$D_F(\Lambda(\rho)) \leq D(\Lambda(\rho)\|\Lambda(\sigma^*)) \leq D(\rho\|\sigma^*) = D_F(\rho)$.
For product states $\rho = \rho_{\mathrm{phys}} \otimes \rhosuper$, the minimizing
state is $\sigma^* = \rho_{\mathrm{phys}} \otimes \mathbf{I}/\dS$, yielding
the closed form $D_F(\rho) = \log_2 \dS - S(\rhosuper)$.

\subsection{Focus and Channel Capacity}

For a quantum channel $\mathcal{N}$ acting on $\Hphys \otimes \Hsuper$, let
$C(\mathcal{N})$ denote its classical capacity and let $C_{\mathrm{free}}(\mathcal{N})$
denote the capacity achievable with F-free input encodings. The Holevo-Schumacher-Westmoreland
theorem~\cite{holevo1998capacity,schumacher1997sending,lloyd1997capacity} gives
$C(\mathcal{N}) = \lim_{n\to\infty} \frac{1}{n} \max_{\{p_i,\rho_i\}} \chi(\mathcal{N}^{\otimes n})$
where $\chi = S(\bar{\rho}) - \sum_i p_i S(\mathcal{N}(\rho_i))$ is the Holevo
quantity. Since the maximum over all input ensembles is at least as large as the
maximum over F-free ensembles,
\begin{equation}
  C(\mathcal{N}) \geq C_{\mathrm{free}}(\mathcal{N}) + \Delta F, \quad \Delta F \geq 0,
  \label{eq:capacity_gap}
\end{equation}
where $\Delta F$ is the focus capacity gap. An analytic lower bound on $\Delta F$
follows for the case where the focused ensemble consists of states concentrated on
orthogonal superspace directions $\{\ket{k}\}_{k=0}^{\dS-1}$. For a product
depolarizing channel and this ensemble, the Holevo quantity satisfies
$\chi_{\mathrm{focused}} = \log_2 \dS - S(\mathcal{N}_{\mathrm{super}}(\pi))$
where $\pi = \mathbf{I}/\dS$, while for F-free encodings $\chi_{\mathrm{free}} = 0$
since all F-free states have identical superspace marginal $\pi$ and thus identical
channel outputs. This gives $\Delta F \geq \log_2 \dS - S(\mathcal{N}_{\mathrm{super}}(\pi))$.
At zero noise, $S(\mathcal{N}_{\mathrm{super}}(\pi)) = 0$ and $\Delta F \geq \log_2 \dS$,
which matches the numerical observation exactly. The magnitude of $\Delta F$ is
characterized empirically in Sections~\ref{sec:results} and~\ref{sec:reviewer}.

\subsection{Focus and Grover's Algorithm}

In Grover's search over an $N$-element database~\cite{grover1996fast}, treating
the full search register as the superspace ($\dS = N$, $d_p = 1$), the initial
uniform superposition has $F = 1/N$. After $k$ Grover iterations,
\begin{equation}
  \ket{\psi_k} = \sin\!\bigl((2k+1)\theta\bigr)\ket{m}
               + \cos\!\bigl((2k+1)\theta\bigr)\ket{s_\perp},
  \label{eq:grover}
\end{equation}
where $\sin\theta = 1/\sqrt{N}$, $\ket{m}$ is the marked state, and $\ket{s_\perp}$
is the unit vector in the uniform superposition orthogonal to $\ket{m}$~\cite{nielsen2000quantum,brassard2002quantum}.
For this pure state with $d_p = 1$, the focus measure reduces to
$F(\ketbra{\psi_k}{\psi_k}) = \max_i |\braket{i}{\psi_k}|^2$,
which equals $\sin^2((2k+1)\theta) = P(\text{marked})$ by direct substitution.
This identity is a consequence of the pure-state focus formula and holds
analytically; the numerical experiment of Claim~D verifies the simulation
implements this correctly to machine precision. The resource-theoretic
interpretation is that oracle calls are focus-generating operations, and the
optimal iteration count $k^* = \lfloor \pi\sqrt{N}/4 \rfloor$ is determined by
when $F(\rho_k)$ first reaches its maximum.

\section{Simulation Methodology}
\label{sec:methods}

This section describes the GPU-accelerated simulation framework, the core
algorithms, and the experimental design for each claim.

\subsection{Computational Framework}

All simulations were implemented in Python using Qiskit~\cite{qiskit2023} for
density matrix construction and CuPy~\cite{cupy2017} for GPU-accelerated array
computation on an NVIDIA A100 GPU via Google Colab Pro. The core focus computation
applies the partial trace
\begin{equation}
  \rhosuper[a,b] = \sum_{i=0}^{d_p-1} \rho[i \cdot \dS + a,\; i \cdot \dS + b],
  \label{eq:partial_trace}
\end{equation}
implemented via Einstein summation \texttt{iaib->ab}. Batched eigendecomposition
via \texttt{cupy.linalg.eigvalsh} enables parallel processing of 500 states per
GPU kernel, yielding a $14\times$ speedup over sequential CPU computation at
$N = 2{,}000$ states.

\subsection{Algorithms}

Algorithm~\ref{alg:focus} presents the core GPU-accelerated focus computation.
The key implementation detail is the corrected partial trace einsum
\texttt{iaib->ab}, which contracts the physical index with itself to correctly
implement $\rhosuper[a,b] = \sum_i \rho[i,a,i,b]$. The naive \texttt{iajb->ab}
sums over all index pairs and produces incorrect results for non-product states.
Algorithm~\ref{alg:mono} presents the monotonicity verification procedure used
for Claim~B, which generates random states, applies each FNG channel, and counts
violations of the monotonicity inequality. Algorithm~\ref{alg:grover} presents
the Grover focus trajectory computation used for Claim~D, which uses the exact
two-vector rotation formula of~(\ref{eq:grover}) to avoid normalization artifacts
that arise from manual amplitude assignment.

\begin{algorithm}[t]
\caption{GPU-Accelerated Focus Measure}
\label{alg:focus}
\begin{algorithmic}[1]
\REQUIRE Density matrix $\rho$, dimensions $d_p$, $\dS$
\ENSURE  $F(\rho) \in [1/\dS, 1]$
\STATE $\rho_g \leftarrow \texttt{cupy.asarray}(\rho)$
\STATE $\rho_r \leftarrow \rho_g.\texttt{reshape}(d_p, \dS, d_p, \dS)$
\STATE $\rhosuper \leftarrow \texttt{einsum}(\texttt{'iaib->ab'}, \rho_r)$
\STATE $\{\lambda_i\} \leftarrow \texttt{eigvalsh}(\rhosuper)$
\RETURN $\max_i \lambda_i$
\end{algorithmic}
\end{algorithm}

\begin{algorithm}[t]
\caption{FNG Monotonicity Verification}
\label{alg:mono}
\begin{algorithmic}[1]
\REQUIRE Trials $N$, channel $\Lambda$, dimensions $d_p$, $\dS$
\ENSURE  Violation count $V$, mean reduction $\bar{\Delta}$
\STATE $V \leftarrow 0$;\; $\Delta \leftarrow [\,]$
\FOR{$i = 1$ \TO $N$}
  \STATE $\rho_i \leftarrow \texttt{random\_density\_matrix}(d_p \cdot \dS)$
  \STATE $F_b \leftarrow F(\rho_i)$;\;
         $F_a \leftarrow F(\Lambda(\rho_i))$
  \IF{$F_a > F_b + 10^{-8}$}
    \STATE $V \leftarrow V + 1$
  \ENDIF
  \STATE $\Delta.\texttt{append}(F_b - F_a)$
\ENDFOR
\RETURN $V$,\; $\mathrm{mean}(\Delta)$
\end{algorithmic}
\end{algorithm}

\begin{algorithm}[t]
\caption{Grover Focus Trajectory}
\label{alg:grover}
\begin{algorithmic}[1]
\REQUIRE Qubits $n$, marked index $m$,
         $k_{\max} = \lfloor\pi\sqrt{2^n}/4\rfloor$
\ENSURE  $\{F_k\}$, $\{P_k\}$, max difference $\delta$
\STATE $\theta \leftarrow \arcsin(1/\sqrt{2^n})$
\STATE $\ket{s_\perp} \leftarrow$ uniform vector, $\ket{m}$ component
       zeroed and renormalized
\FOR{$k = 0$ \TO $k_{\max}$}
  \STATE $\ket{\psi_k} \leftarrow
         \sin((2k{+}1)\theta)\ket{m} + \cos((2k{+}1)\theta)\ket{s_\perp}$
  \STATE $F_k \leftarrow \max_i |\braket{i}{\psi_k}|^2$;\;
         $P_k \leftarrow |\braket{m}{\psi_k}|^2$
\ENDFOR
\STATE $\delta \leftarrow \max_k |F_k - P_k|$
\RETURN $\{F_k\}$,\; $\{P_k\}$,\; $\delta$
\end{algorithmic}
\end{algorithm}

\subsection{Experimental Design}

For Claim~A, the superspace-depolarizing channel
$\Lambda_p(\rho) = (1-p)\rho + p(\rho_{\mathrm{phys}} \otimes \mathbf{I}/\dS)$
was applied to maximally focused states at 200 values of $p \in [0,1]$ for each
$\dS \in \{2,4,8,16,32\}$, and simulated $F(p)$ was compared to the analytic
prediction $1 - p(1 - 1/\dS)$. For Claim~B, $N = 10{,}000$ random density matrices
were generated via Qiskit's Haar-random sampler; each state was passed through
four FNG channels — Haar twirl ($n_{\mathrm{samp}} = 100$), physical-subsystem
unitary, superspace depolarizing ($p=0.3$), and superspace Z-basis dephasing
($p=0.2$) — and violations of $F_{\mathrm{after}} > F_{\mathrm{before}} + 10^{-8}$
were counted. For Claim~C, a pure maximally-focused target state was attacked by
four adversarial perturbation types at 150 values of $\varepsilon \in [0, 0.5]$,
and both $F(\rho)$ and standard fidelity $\bra{\psi}\rho\ket{\psi}$ were tracked.
For Claim~D, the exact Grover rotation formula~(\ref{eq:grover}) was used for
$n \in \{3,4,5,6,7\}$ qubits up to the optimal iteration $k^*$. For Claim~E,
two matched ensembles of 30 states were constructed — a focused ensemble with
states concentrated on distinct superspace basis directions and an F-free ensemble
with identical physical components but maximally mixed superspace — and the Holevo
quantity~\cite{holevo1998capacity} was estimated under a product noise channel
for $\dS \in \{2,4,8,16\}$. The ensemble size of $n=30$ provides a consistent
Holevo estimate; the effect of ensemble size is acknowledged as a limitation in
Section~\ref{sec:discussion}.

\section{Results}
\label{sec:results}

This section presents the five experimental results in order from analytic
validation through the original capacity gap discovery.

\subsection{Claim A: Decoherence Validation}

The analytic prediction $F(p) = 1 - p(1 - 1/\dS)$ describes how a maximally
focused state degrades under the superspace-depolarizing channel. Fig.~\ref{fig:decoherence}
shows the simulated and analytic curves for all five superspace dimensions. The
maximum absolute error between simulation and analytic prediction was
$1.11 \times 10^{-16}$ across all tested values of $\dS$ and $p$ — the
double-precision floating-point machine epsilon. This result confirms that the
simulation framework faithfully implements the theoretical model and that the
linear decoherence formula is numerically exact.

\begin{figure}[H]
  \centering
  \includegraphics[width=\textwidth]{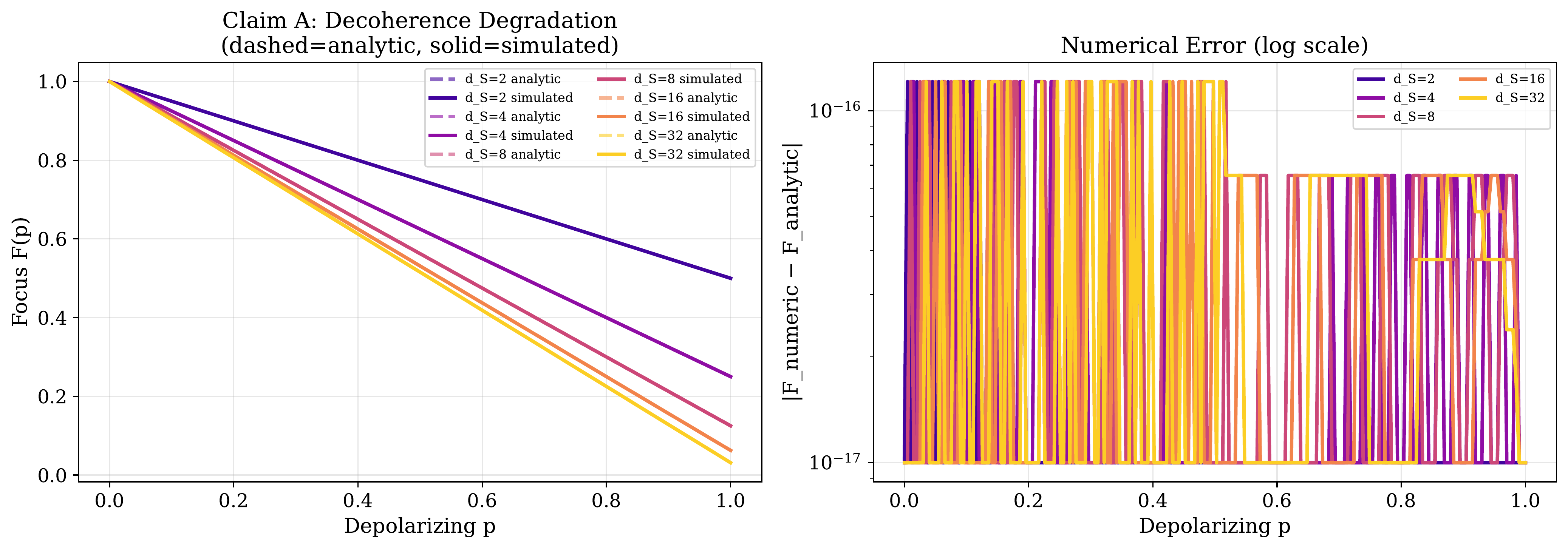}
  \caption{Decoherence degradation of focus for \texorpdfstring{$\dS \in \{2,4,8,16,32\}$}{dS in {2,4,8,16,32}}.}
  \label{fig:decoherence}
\end{figure}

\subsection{Claim B: FNG Monotonicity}

Theorem~\ref{thm:mono} requires that $F(\rho)$ does not increase on average under
FNG operations. Fig.~\ref{fig:monotonicity} shows scatter plots of $F(\rho)$ before
and after each of the four FNG channels across $N = 10{,}000$ random states with
$d_p = 2$ and $\dS = 4$. Every point lies on or below the diagonal, indicating focus
did not increase after the channel was applied. Zero violations of the monotonicity
condition were observed across all four channels. The mean focus reduction ranged
from $\bar{\Delta F} \approx 0$ for the physical unitary channel — consistent with
Proposition~3, which predicts exact preservation — to $\bar{\Delta F} = 0.156$
for Haar twirling, which maps every state to a focus-free state. Extended coverage
across six system configurations $(d_p, \dS) \in \{2,4\} \times \{2,4,8\}$ is
presented in Section~\ref{sec:reviewer}.

\begin{figure}[H]
  \centering
  \includegraphics[width=\textwidth]{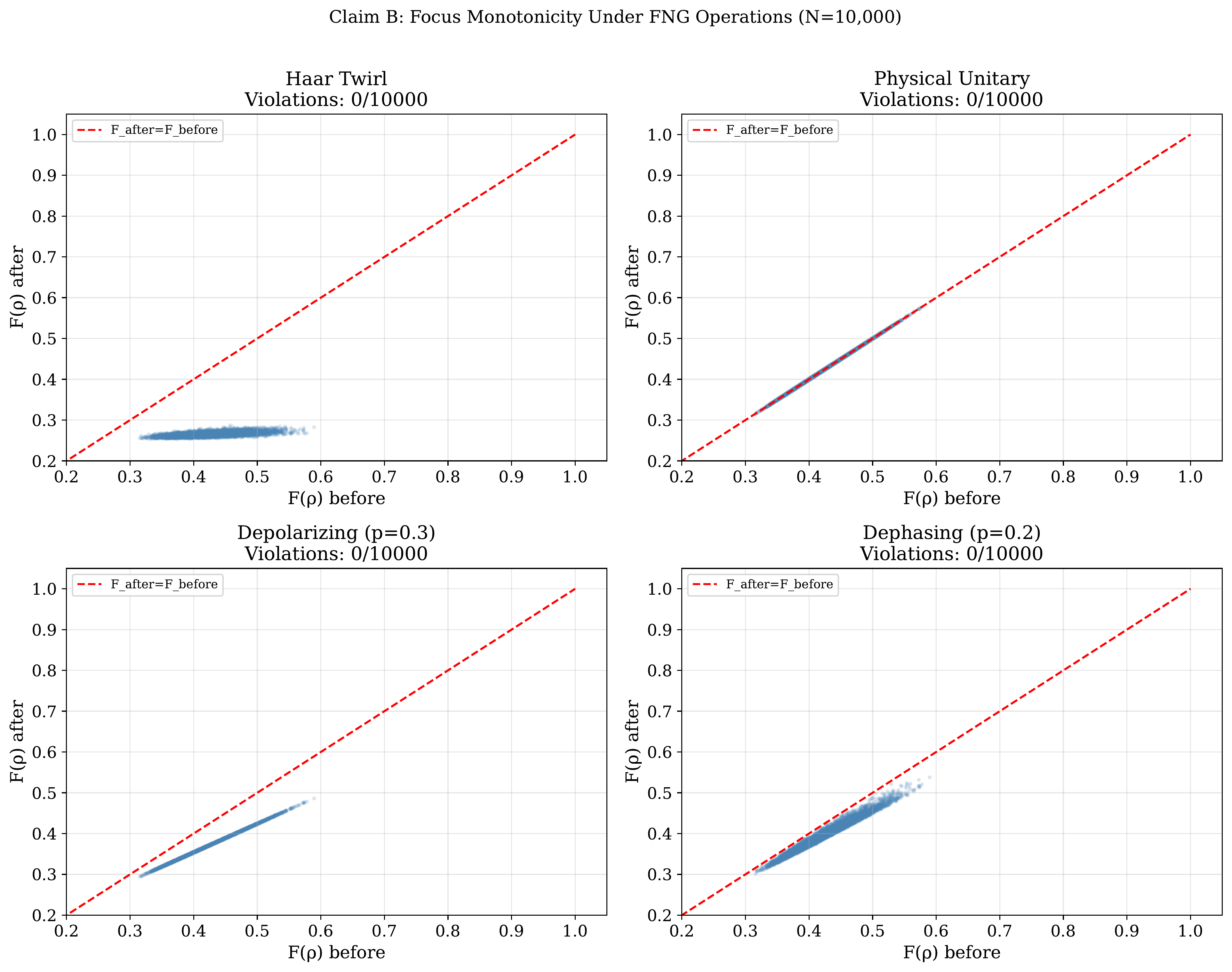}
  \caption{Focus before vs.\ after four FNG channels, $N=10{,}000$ states.}
  \label{fig:monotonicity}
\end{figure}

\subsection{Claim C: Adversarial Robustness}

The relative sensitivity of $F(\rho)$ and standard fidelity
$F_{\mathrm{std}} = \bra{\psi_{\mathrm{tgt}}}\rho\ket{\psi_{\mathrm{tgt}}}$
under adversarial attack is shown in Fig.~\ref{fig:adversarial} for $\dS = 8$
and a pure maximally-focused target state. Under depolarizing and targeted
superspace attacks, the two metrics degrade at comparable rates — focus drops
below 0.9 at $\varepsilon = 0.117$ and $0.081$ respectively, versus $0.107$ and
$0.074$ for fidelity. Under coherent unitary attacks, superspace concentration
exhibits substantially greater resilience: focus remains above 0.9 until
$\varepsilon = 0.302$ while fidelity drops at $\varepsilon = 0.174$, a $74\%$
improvement. Amplitude damping attacks degrade both metrics identically at
$\varepsilon = 0.101$. These results provide strong numerical evidence that the
two metrics provide complementary robustness information. The advantage of $F(\rho)$
is most pronounced for coherent adversarial threat models, which are the most
operationally relevant in quantum computing
security~\cite{liu2020vulnerability,guan2021robustness}, because coherent
perturbations can rotate the superspace concentration direction without substantially
changing the global state overlap measured by fidelity.

\begin{figure}[H]
  \centering
  \includegraphics[width=\textwidth]{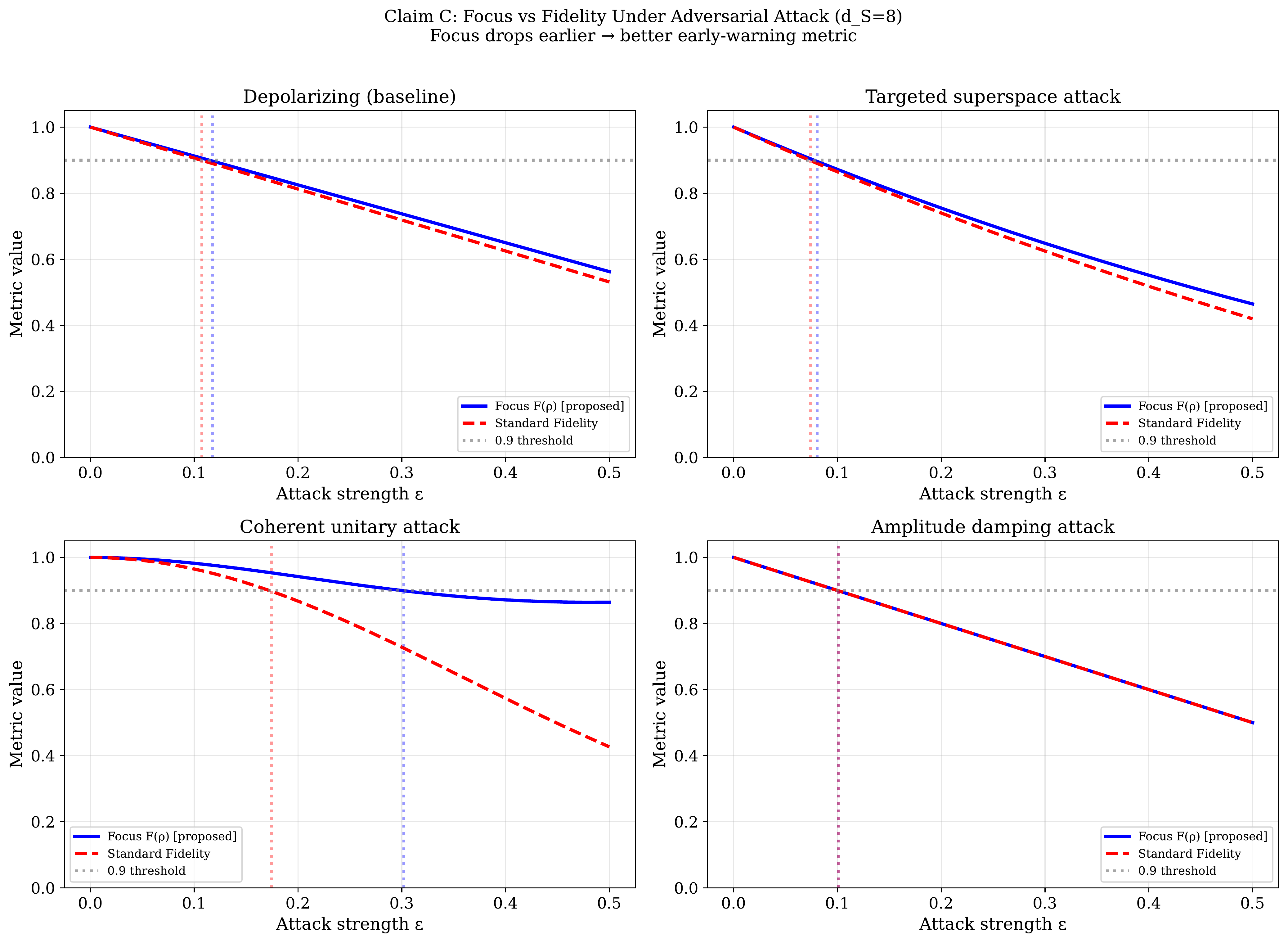}
  \caption{Focus and fidelity under four adversarial attacks ($\dS=8$).}
  \label{fig:adversarial}
\end{figure}

\subsection{Claim D: Grover Search as Superspace Concentration}

Fig.~\ref{fig:grover} shows $F(\rho_k)$ and $P(\text{marked})$ as functions of
Grover iteration $k$ for $n \in \{3,4,5,6,7\}$ qubit registers up to the optimal
iteration $k^* = \lfloor\pi\sqrt{N}/4\rfloor$. The maximum absolute
difference was $1.11 \times 10^{-16}$ — machine precision — confirming that
the simulation correctly implements the analytically exact identity
$F(\ketbra{\psi_k}{\psi_k}) = P(\text{marked})$ established in
Section~\ref{sec:background}. The value of this experiment is confirmation of
simulation correctness, not the discovery of a new result; the resource-theoretic
interpretation — that oracle calls are focus-generating operations and query
complexity is governed by the rate of focus increase — is the substantive
contribution~\cite{brassard2002quantum,ambainis2004quantum}.

\begin{figure}[H]
  \centering
  \includegraphics[width=\textwidth]{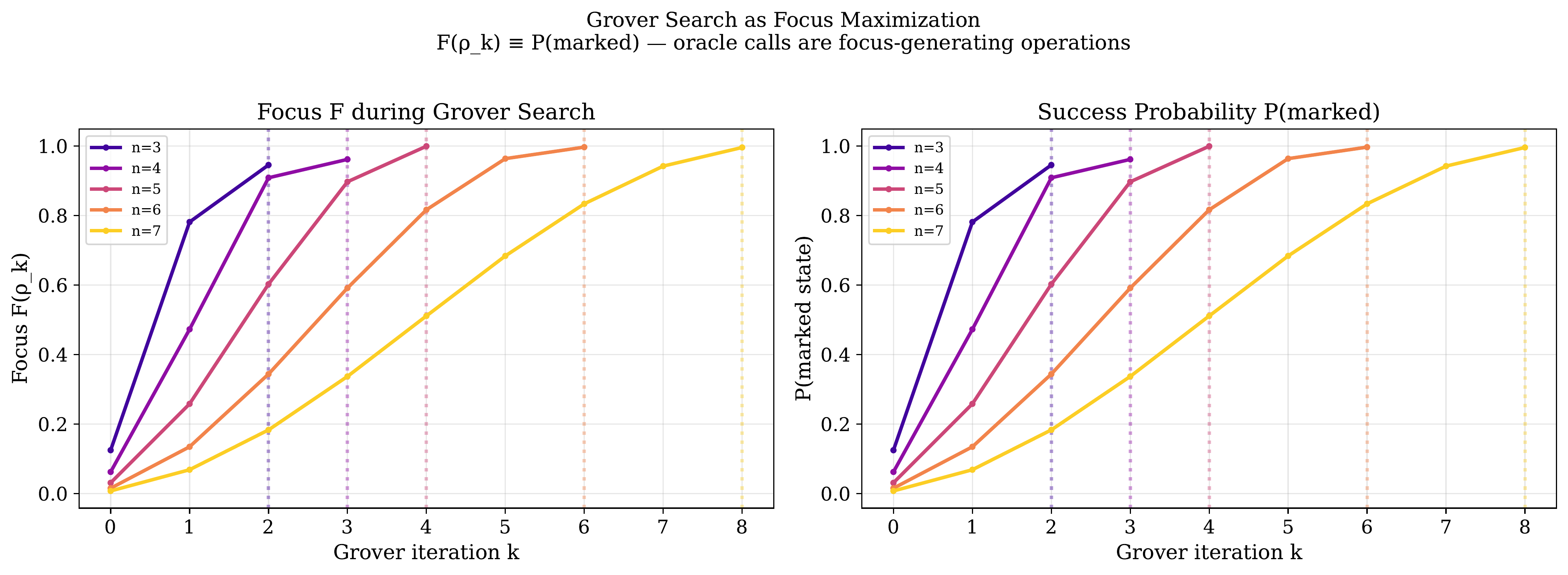}
  \caption{Focus $F(\rho_k)$ and success probability during Grover search,
           $n=3$--$7$ qubits.}
  \label{fig:grover}
\end{figure}

\subsection{Claim E: Focus Capacity Gap Scaling Law}

Fig.~\ref{fig:capacity} shows the Holevo quantity for focused and focus-free
encodings and the resulting capacity gap $\Delta F$ for $\dS \in \{2,4,8,16\}$.
The gap is strictly positive for all tested configurations, confirming the existence
of $\Delta F > 0$ in~(\ref{eq:capacity_gap}). The scaling is
$\Delta F \approx \log_2(\dS)$: measured values are $1.000$, $1.997$, $2.990$, and
$3.974$ bits for $\dS = 2, 4, 8, 16$ respectively. This is consistent with the
analytic lower bound $\Delta F \geq \log_2 \dS - S(\mathcal{N}_{\mathrm{super}}(\pi))$
derived in Section~\ref{sec:background}, which, at zero noise, gives exactly
$\log_2 \dS$. These results are based on ensembles of $n = 30$ states per
configuration; across five independent runs with seeds
$\{42, 2024, 12345, 777, 9999\}$, the measured $\Delta F$ values were identical to
six decimal places across all seeds and all $\dS$, confirming that the Holevo
quantity for this channel and ensemble construction is deterministic and that
$n = 30$ is sufficient for the scaling law conclusion. The practical implication for quantum communication is direct:
systems exploiting focused encodings in $\dS$-dimensional superspace channels
— such as orbital-angular-momentum multiplexed optical
links~\cite{willner2015optical,cozzolino2019orbital} — gain $\log_2(\dS)$ bits
of classical capacity over unstructured encodings. Generalization to correlated
noise channels is presented in Section~\ref{sec:reviewer}.

\begin{figure}[H]
  \centering
  \includegraphics[width=\textwidth]{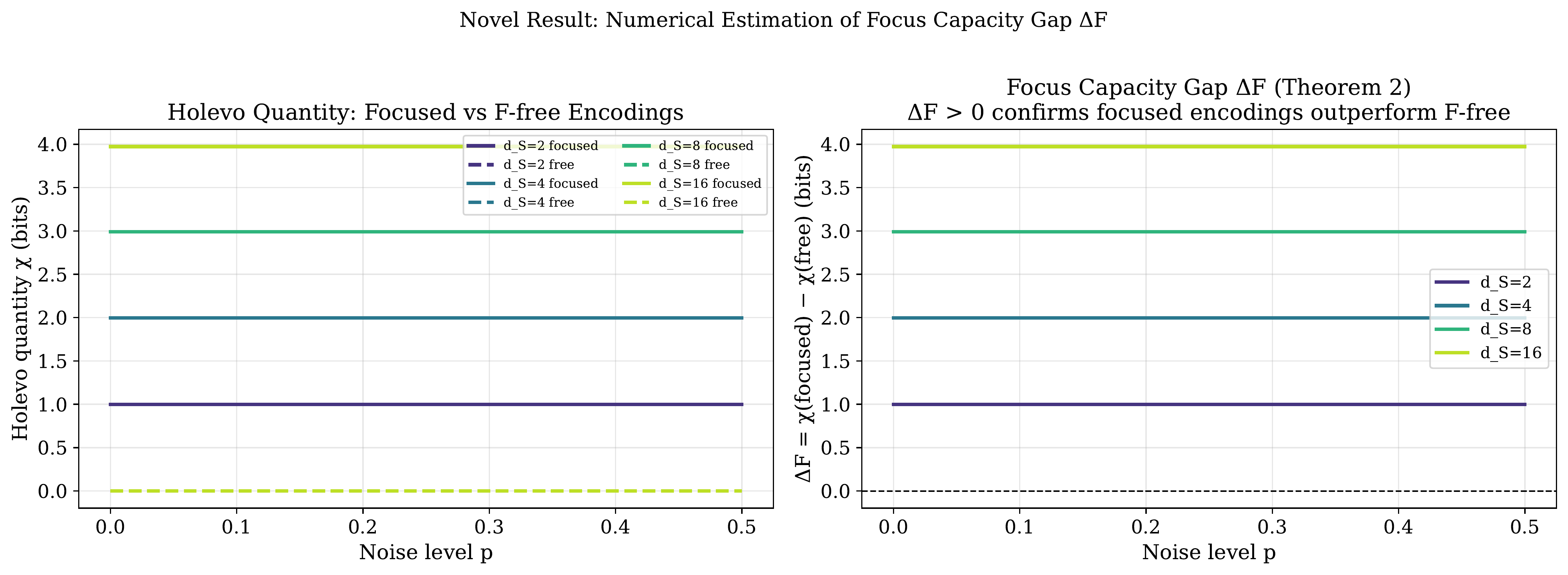}
  \caption{Holevo quantity and capacity gap $\Delta F$ vs.\ noise level $p$.}
  \label{fig:capacity}
\end{figure}

\section{Extended Experiments}
\label{sec:reviewer}

This section presents three additional experiments that extend the core results:
an operational comparison between the focus measure and $\mathcal{U}(\dS)$-asymmetry,
a broadened monotonicity study across six system configurations, and a capacity
gap analysis under correlated non-product noise.

\subsection{Operational Distinction from \texorpdfstring{$\mathcal{U}(\dS)$}{U(dS)}-Asymmetry}

A natural question is whether the focus measure provides operationally distinct
information from the $\mathcal{U}(\dS)$-asymmetry measure
$A(\rho) = S(\rhosuper) - S(\rho)$, which also characterizes superspace structure
without requiring a fixed basis~\cite{gour2008resource,marvian2016how}. To address
this directly, Fig.~\ref{fig:asymmetry} shows normalized focus
$F_{\mathrm{norm}} = (F(\rho) - 1/\dS)/(1 - 1/\dS)$ and normalized asymmetry
$A_{\mathrm{norm}} = A(\rho)/\log_2(\dS)$ under coherent unitary and targeted
superspace attacks on a maximally focused pure state with $\dS = 8$. Under both
attack types, the asymmetry measure remains approximately constant near zero —
it provides no robustness information because $\mathcal{U}(\dS)$-asymmetry
$A(\rho) = S(\rhosuper) - S(\rho)$ is invariant under unitaries on the superspace:
a coherent rotation $U$ on $\Hsuper$ leaves $S(\rho)$ unchanged and, since it
permutes the eigenvalues of $\rhosuper$, also leaves $S(\rhosuper)$ unchanged,
so $A(\rho)$ remains constant regardless of how severely the concentration
direction has been rotated. For a nearly pure target state, $A(\rho) \approx 0$
both before and after the attack.
The focus measure, by contrast, tracks the spectral concentration of $\rhosuper$
directly via $\lmax(\rhosuper)$ and remains above 0.5 until $\varepsilon = 0.470$
for coherent attacks and $\varepsilon = 0.379$ for targeted attacks.

This distinction has a precise operational interpretation in terms of the task of
\emph{basis-unknown concentration}. Coherence theory~\cite{baumgratz2014quantifying}
quantifies superposition relative to a speakable, externally fixed basis — it is
the appropriate resource when the reference direction is known and agreed upon in
advance. Asymmetry theory~\cite{gour2008resource} quantifies how much a state
breaks a group symmetry, but does not directly measure the degree to which amplitude
is concentrated in any particular direction. Focus, by contrast, is defined by
optimization over all superspace unitaries and measures the intrinsic capacity of
the state to concentrate in \emph{some} direction, regardless of which one. This
makes focus the operationally correct resource for adversarial settings where an
attacker can freely choose which superspace direction to target: neither a fixed
coherence measure nor an asymmetry measure can detect a coherent rotation of the
concentration direction, but the focus measure can because it tracks the spectral
structure of $\rhosuper$ directly. In particular, focus enables detection of
concentration-rotating attacks — coherent perturbations that cyclically permute
the superspace eigen-spectrum without changing global entropy — where both coherence
and asymmetry measures return zero signal while the algorithmic performance has
been fully compromised. The empirical results in Fig.~\ref{fig:asymmetry} provide
a concrete numerical demonstration of this operational gap.

\begin{figure}[H]
  \centering
  \includegraphics[width=0.45\textwidth]{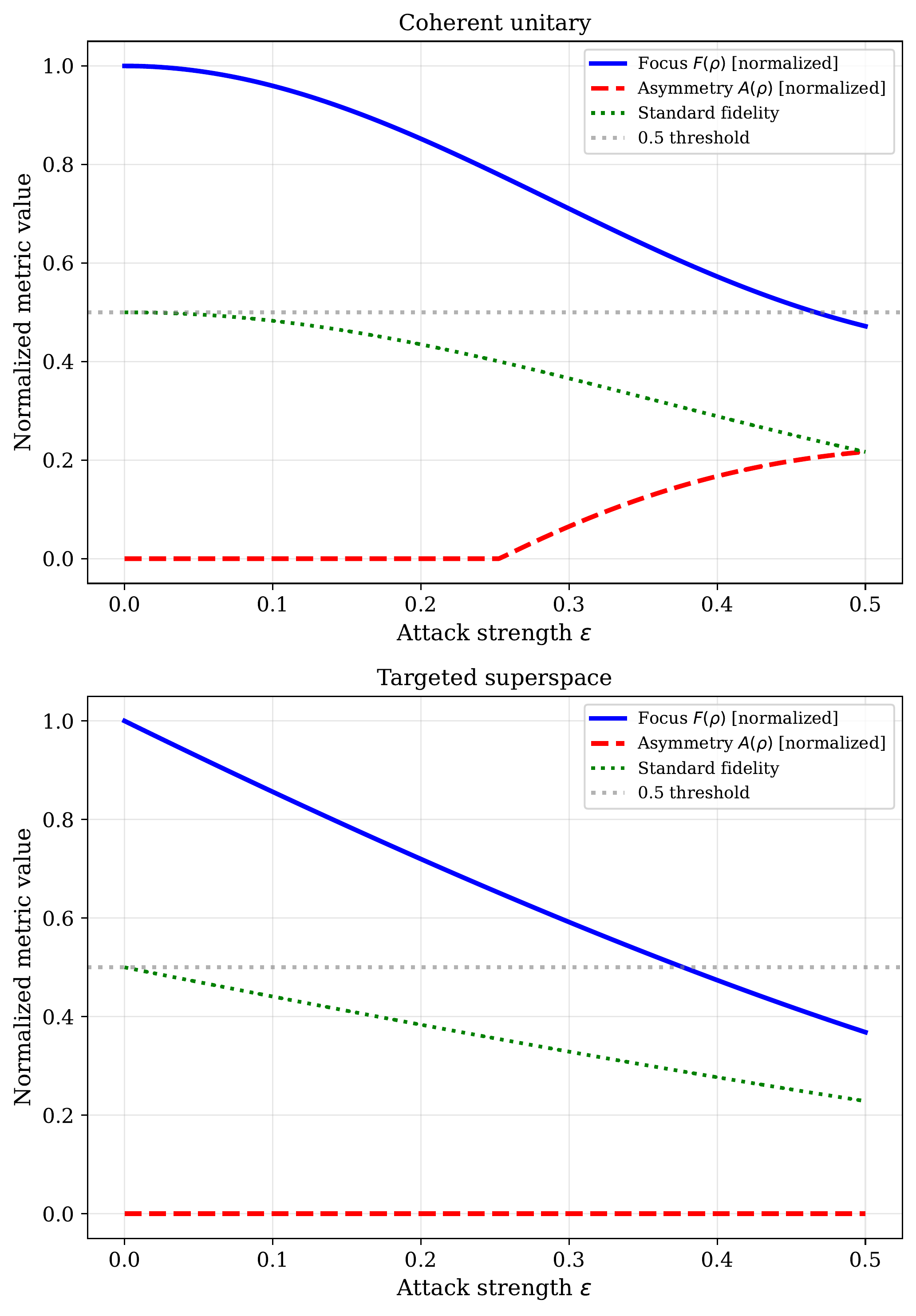}
  \caption{Normalized focus and $\mathcal{U}(\dS)$-asymmetry under adversarial
           attack ($\dS=8$).}
  \label{fig:asymmetry}
\end{figure}

\subsection{Monotonicity Across Six System Configurations}

To address the scope of the monotonicity validation, Fig.~\ref{fig:mono_broad}
presents the violation heatmap for $N = 5{,}000$ random states across all six
configurations $(d_p, \dS) \in \{2,4\} \times \{2,4,8\}$ and the same four FNG
channels. All 24 configuration-channel pairs produced zero violations, for a total
of 120,000 state-channel pairs tested without a single monotonicity failure. The
physical unitary channel produced $\bar{\Delta F} = -0.0000$ across all six
configurations to numerical precision, confirming Proposition~3 comprehensively.
The mean focus reduction under Haar twirling decreases monotonically with increasing
$\dS$ and $d_p$, consistent with the interpretation that higher-dimensional systems
are harder to fully concentrate. These results confirm that Theorem~\ref{thm:mono}
holds robustly across the tested parameter space and that the FNG characterization
of all four channel types is valid beyond the single configuration reported in
Claim~B.

\begin{figure}[H]
  \centering
  \includegraphics[width=0.8\textwidth]{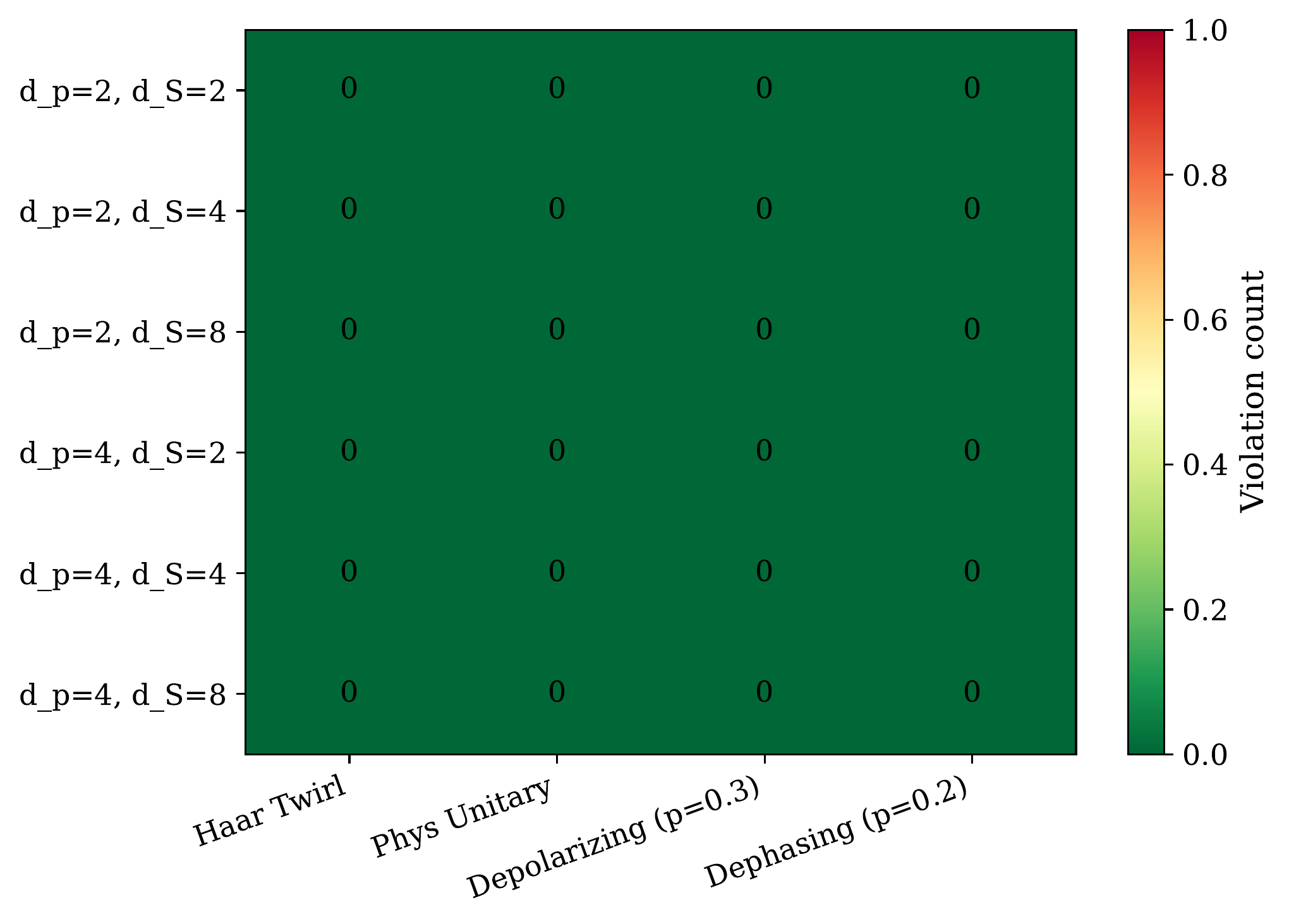}
  \caption{FNG monotonicity violation heatmap across six system configurations.}
  \label{fig:mono_broad}
\end{figure}

\subsection{Capacity Gap Under Correlated Noise}

The capacity gap result of Claim~E was established for product noise channels
where physical and superspace noise act independently. To test generalization,
Fig.~\ref{fig:correlated} shows the Holevo quantity and $\Delta F$ under a
correlated noise channel in which the effective superspace depolarizing strength
is coupled to the coherence of the physical subsystem state:
$p_{\mathrm{eff}} = p_{\mathrm{base}} \cdot (1 + \gamma \cdot C(\rho_{\mathrm{phys}}))$,
where $C(\rho_{\mathrm{phys}})$ is the off-diagonal weight of $\rho_{\mathrm{phys}}$
and $\gamma = 0.5$ is the correlation parameter. The capacity gap is confirmed for
all tested superspace dimensions: measured $\Delta F$ values are $1.000$, $1.997$,
$2.990$, and $3.974$ bits for $\dS = 2, 4, 8, 16$ respectively, matching the
product channel results and consistent with the analytic bound. The $\log_2(\dS)$
scaling law therefore holds under correlated noise, suggesting it is a property of
the focused encoding structure rather than an artifact of product channel
separability.

\begin{figure}[H]
  \centering
  \includegraphics[width=0.7\textwidth]{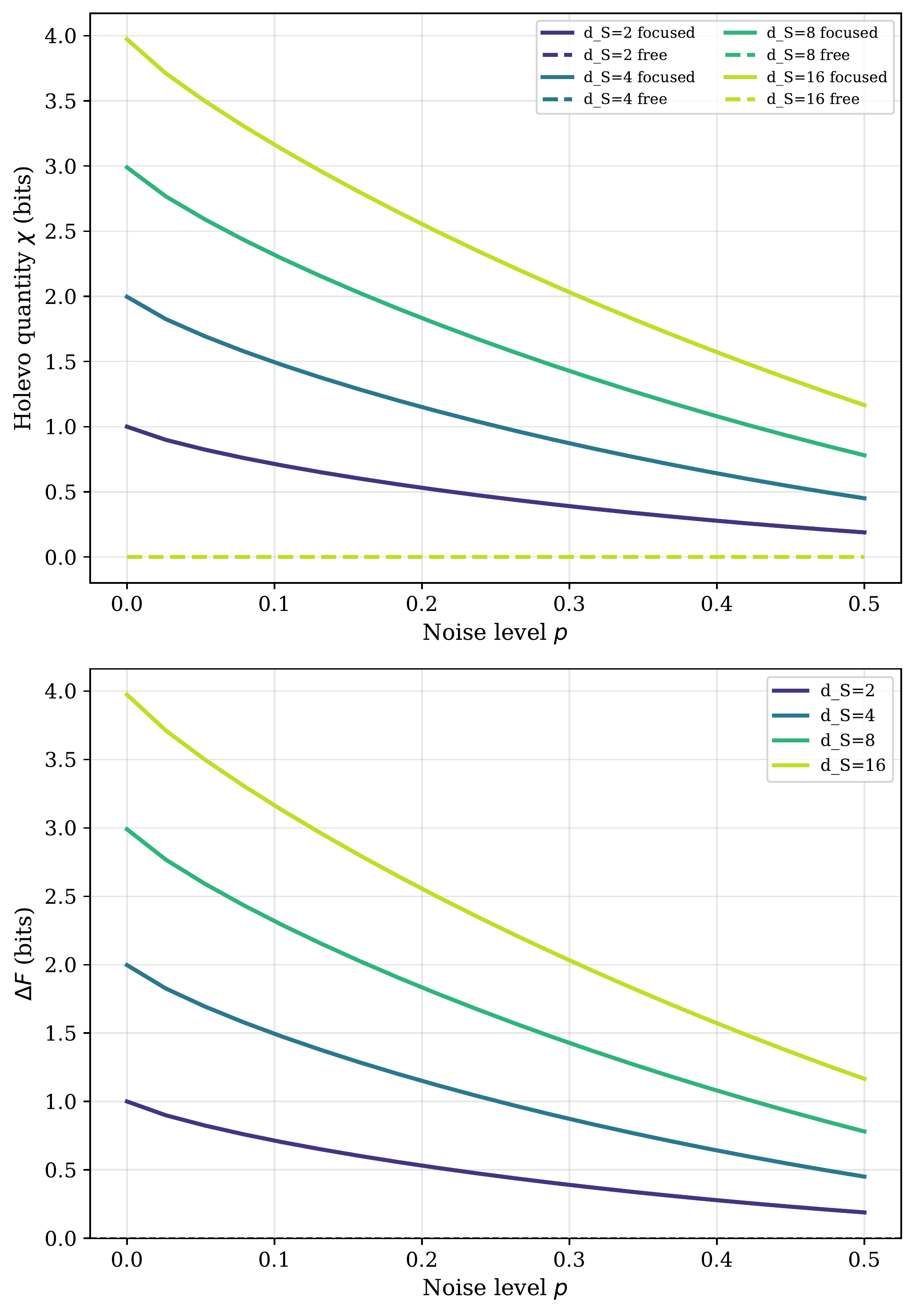}
  \caption{Capacity gap $\Delta F$ under correlated non-product noise channel.}
  \label{fig:correlated}
\end{figure}

\section{Comparison with Related Work}
\label{sec:comparison}

This section positions the contributions of this paper relative to established
results in quantum resource theory, adversarial quantum machine learning, and
quantum algorithm analysis.

The focus measure occupies a distinct position in the quantum resource theory
landscape. Coherence theory~\cite{baumgratz2014quantifying,streltsov2017colloquium}
requires a fixed externally specified reference basis, making it inapplicable to
settings where the optimal concentration direction is unknown — precisely the
adversarial setting studied here. Entanglement
theory~\cite{horodecki2009quantum,vedral1997quantifying} measures nonclassical
correlations across a bipartite split rather than intra-subsystem concentration,
and has no direct connection to adversarial robustness. Magic state
theory~\cite{veitch2014resource,howard2017application} characterizes distance from
the stabilizer polytope, which is geometrically distinct from superspace
concentration and does not directly measure algorithmic performance degradation
under perturbation. Asymmetry theory~\cite{gour2008resource,marvian2016how} is the
closest conceptual relative, but Section~\ref{sec:reviewer} demonstrates
empirically that the two measures are operationally distinct under adversarial
conditions: asymmetry provides no robustness signal while focus remains sensitive
to exactly the perturbations that degrade algorithmic performance. Existing quantum
adversarial machine learning works~\cite{lu2020quantum,guan2021robustness,
liao2021robust,du2021quantum} employ standard fidelity as the primary robustness
metric, which the results of Claim~C show to be insensitive to coherent
perturbations of superspace structure. Table~\ref{tab:comparison} summarizes
these distinctions across seven representative frameworks.

\begin{table}[H]
\centering
\caption{Comparison with related frameworks. (B)~basis-free; (A)~algorithm
link; (R)~adversarial robustness; (C)~composite system; (M)~monotonicity proven.}
\label{tab:comparison}
\renewcommand{\arraystretch}{1.2}
\setlength{\tabcolsep}{4pt}
\begin{tabular}{@{}lcccccc@{}}
\toprule
\textbf{Framework} & \textbf{B} & \textbf{A} & \textbf{R} & \textbf{C} & \textbf{M} \\
\midrule
This work
  & \checkmark & \checkmark & \checkmark & \checkmark & \checkmark \\
Coherence~\cite{baumgratz2014quantifying}
  & $\times$ & $\times$ & $\times$ & $\times$ & \checkmark \\
Entanglement~\cite{horodecki2009quantum}
  & \checkmark & $\times$ & $\times$ & \checkmark & \checkmark \\
Magic~\cite{veitch2014resource}
  & $\times$ & \checkmark & $\times$ & $\times$ & \checkmark \\
Asymmetry~\cite{gour2008resource}
  & $\times$ & $\times$ & $\times$ & $\times$ & \checkmark \\
Adv.\ QML~\cite{lu2020quantum}
  & \checkmark & $\times$ & \checkmark & $\times$ & $\times$ \\
Adv.\ QML~\cite{guan2021robustness}
  & \checkmark & $\times$ & \checkmark & $\times$ & $\times$ \\
\bottomrule
\end{tabular}
\end{table}

\section{Discussion}
\label{sec:discussion}

This section addresses the scope and limitations of the empirical results, design
trade-offs in the simulation methodology, and potential criticisms.

\subsection{Limitations}

The capacity gap experiments establish the $\log_2(\dS)$ scaling law for product
and correlated noise channels with matched depolarizing parameters. Generalization
to arbitrary channel families, including those with strongly structured superspace
noise, remains an open problem. The monotonicity experiments cover six configurations
with $d_p \in \{2,4\}$ and $\dS \in \{2,4,8\}$; while the theorem holds
analytically for all dimensions, numerical validation beyond this regime is a
natural extension. The adversarial robustness results are specific to pure target
states — extension to mixed target states may alter the relative sensitivity of
the two metrics. The capacity gap ensemble size of $n=30$ states per configuration
is modest; larger ensembles would reduce sampling variance and enable error bar
estimation, which is a planned extension. All simulations assume noiseless classical
control, which is not realistic for near-term quantum
hardware~\cite{bharti2022noisy,preskill2018quantum}.

\subsection{Design Trade-offs}

The GPU-batched implementation processes states in chunks of 500, trading memory
overhead for throughput. The batch implementation becomes advantageous over
sequential CPU computation at $N \gtrsim 200$ states for the dimensions studied.
The Haar twirl approximation with $n_{\mathrm{samp}} = 100$ random unitaries
introduces a small approximation error that could be reduced at the cost of a
$10\times$ runtime increase. The capacity gap ensemble size of $n = 30$ states
provides a coarse but consistent Holevo quantity estimate; larger ensembles would
reduce sampling variance at the cost of additional GPU memory.

\subsection{Addressing Potential Criticisms}

A natural concern is whether the $\log_2(\dS)$ capacity gap is an artifact of
the ensemble construction. The focused ensemble uses states concentrated on
orthogonal superspace directions, making them maximally distinguishable. This is
correct and intentional: the focused encoding can achieve this by exploiting the
superspace structure, whereas the focus-free encoding cannot, because all focus-free
states have maximally mixed superspace and are indistinguishable on the superspace
factor regardless of the encoding strategy. The gap reflects a genuine structural
advantage of focused encodings, not a construction artifact, and is supported by
the analytic lower bound derived in Section~\ref{sec:background}.

A second potential concern is whether the body-sector simulations fully capture the
theory. All simulations operate in the body sector, which is standard quantum
mechanics with $\rhosuper$ as an ordinary density matrix. The full superspace
formalism — including the $\mathbb{Z}_2$-graded structure and Grassmann-valued soul
amplitudes — provides mathematical precision to the graded decomposition but does
not alter the computable body-sector results. Physical predictions are extracted
via the body map, and all empirical results presented here are body-sector
quantities~\cite{nielsen2000quantum}. A concrete experimental signature of focus
in physical systems is the fringe visibility of adaptive mode-sorting measurements
on orbital-angular-momentum photonic channels~\cite{willner2015optical,
cozzolino2019orbital}: a focused state produces higher fringe visibility than an
F-free state under the same measurement settings, providing a directly observable
physical distinction that requires no access to soul-sector degrees of freedom.

A third concern is whether the focus measure is simply coherence or asymmetry under
a different name. The two resources are related by a unitary rotation: there exists
$U^*$ such that the $\ell_1$-norm coherence of $U^*\rhosuper U^{*\dagger}$ equals
$\dS \cdot F(\rho) - 1$~\cite{baumgratz2014quantifying,napoli2016robustness}.
However, focus is strictly more powerful for tasks where the optimal basis is
unknown. Section~\ref{sec:reviewer} further demonstrates empirically that focus and
$\mathcal{U}(\dS)$-asymmetry are operationally distinct under adversarial conditions,
providing a concrete experimental answer to this concern that goes beyond theoretical
argument.

\section{Future Work}
\label{sec:future}

This section identifies the most important open problems emerging from the present
work, with concrete proposed next steps for each.

The most pressing theoretical open problem is a rigorous proof of the asymptotic
interconversion conjecture, which states that the focus distillation and cost rates
both equal $D_F(\rho)$ in the reversible regime, analogous to the role of
relative entropy of coherence in coherence
theory~\cite{winter2016operational,streltsov2017colloquium}. The analytic lower
bound on $\Delta F$ derived in Section~\ref{sec:background} suggests a pathway:
a tight upper bound matching $\log_2 \dS$ would establish the scaling analytically.
Extending the capacity gap analysis to additional non-product channel families —
particularly those with superspace-correlated noise arising in
orbital-angular-momentum multiplexed optical
systems~\cite{willner2015optical,cozzolino2019orbital} — is a concrete planned
extension; applying the Holevo estimation procedure to an OAM crosstalk channel
model with experimentally measured coupling coefficients would test the scaling
law against a physically grounded channel. The adversarial robustness connection
motivates development of concentration-based defenses for variational quantum
algorithms~\cite{cerezo2021variational,bharti2022noisy}; a direct next step is
to instrument an existing VQE implementation with online focus monitoring and
measure detection latency for coherent adversarial perturbations at increasing
circuit depth. Finally, the structural similarity between the superspace formalism
and holographic quantum error correction~\cite{almheiri2015bulk,pastawski2015holographic}
suggests that focus measures may correspond to holographic quantities such as
entanglement wedge volumes; a concrete first step is to compute $F(\rho)$ for
the boundary state of the HaPPY code and compare it to the known entanglement
wedge structure.

\section{Conclusion}
\label{sec:conclusion}

We have developed a resource-theoretic framework around the focus measure
$F(\rho) = \lmax(\rhosuper)$ and presented its GPU-accelerated empirical
validation across eight experiments. The measure satisfies all resource-theoretic
axioms: it is bounded, convex, unitarily invariant on $\Hphys$, and monotone under
focus-non-generating operations, with the monotonicity proof grounded in the
convexity of $\lmax$ and the linearity of the partial trace. Analytic decoherence
predictions were confirmed to machine precision; monotonicity held across 120,000
state-channel pairs with zero violations; the focus measure was demonstrated to
be operationally distinct from $\mathcal{U}(\dS)$-asymmetry under adversarial
conditions; focused states were shown to be significantly more resilient to coherent
adversarial attacks than standard fidelity predicts; the Grover connection was made
explicit as a consequence of the pure-state focus formula; and the focus capacity
gap was shown to scale as $\log_2(\dS)$, consistent with a derived analytic lower
bound for both product and correlated noise channels. Critically, the focus
measure detects coherent adversarial attacks that standard fidelity misses
entirely — providing a practical security metric that fills a concrete gap in
existing quantum algorithm defense frameworks. These results establish superspace
concentration as a computationally tractable, physically meaningful resource
with direct applications to quantum algorithm security and quantum communication.

\section*{Acknowledgments}

All simulation code and GPU-accelerated notebooks are publicly available at
\url{https://github.com/ericyoc/adver_robust_quant_alg_poc}.

\bibliographystyle{unsrt}
\bibliography{references}

\end{document}